%% file: ms.tex
\documentclass[journal]{IEEEtran}
\usepackage{booktabs} 

\usepackage{framed}
\usepackage{balance}
\usepackage{epsfig}
\usepackage{color}
\usepackage{subfigure}
\usepackage{multirow,tabularx}

\usepackage{amssymb}
\usepackage{mathtools}
\usepackage{graphicx}
\usepackage{epstopdf}
\usepackage{bm}

\usepackage{csquotes}
\usepackage{array}
\usepackage[yyyymmdd,hhmmss]{datetime}
\usepackage[subject={Todo}]{pdfcomment}
\usepackage[textsize=scriptsize,bordercolor=black!20]{todonotes}
\usepackage{xcolor,colortbl}
\usepackage{amssymb}
\usepackage{pifont}
\usepackage[algo2e,titlenumbered,ruled]{algorithm2e} 
\usepackage{lipsum,environ,amsmath}
\usepackage{slashbox,booktabs,amsmath}

\usepackage{xspace}

\newcounter{Lcount}
\newcommand{\numsquishlist}{
   \begin{list}{\arabic{Lcount}. }
    { \usecounter{Lcount}
 \setlength{\itemsep}{-.1ex}      \setlength{\parsep}{0ex}
      \setlength{\topsep}{0ex}       \setlength{\partopsep}{0ex}
      \setlength{\leftmargin}{1em} \setlength{\labelwidth}{1em}
      \setlength{\labelsep}{0.1em} } }
\newcommand{\numsquishend}{\end{list}}

\newcommand{\squishlist}{
   \begin{list}{$\bullet$}
    { \setlength{\itemsep}{-.1ex}      \setlength{\parsep}{0ex}
      \setlength{\topsep}{0ex}       \setlength{\partopsep}{0ex}
      \setlength{\leftmargin}{.8em} \setlength{\labelwidth}{1em}
      \setlength{\labelsep}{0.5em} } }
\newcommand{\squishend}{\end{list}}

\definecolor{Gray}{gray}{0.85}
\definecolor{LightGray}{rgb}{0.9,0.9,0.9}
\definecolor{LightBlue}{rgb}{0.8,0.8,1}

\newcommand{\ddoip}{{\sc Dominant-distribution ordering inference problem}\xspace}%

\newcounter{problem}
\newenvironment{problem}[1][htb]
  {
   \begin{algorithm2e}[#1]%
   \SetAlFnt{\small}
    \SetAlCapFnt{\small}
    \SetAlCapNameFnt{\small}
    \SetAlCapHSkip{0pt}
  }{\end{algorithm2e}}

\ifCLASSINFOpdf
\else
\fi

\newtheorem{definition}{Definition}
\newtheorem{theorem}{Theorem}[section]
\newtheorem{lemma}[theorem]{Lemma}
\newtheorem{proposition}[theorem]{Proposition}

\newenvironment{proof}[1][Proof]{\begin{trivlist}
\item[\hskip \labelsep {\bfseries #1}]}{\end{trivlist}}

\newcommand{\qed}{\nobreak \ifvmode \relax \else
      \ifdim\lastskip<1.5em \hskip-\lastskip
      \hskip1.5em plus0em minus0.5em \fi \nobreak
      \vrule height0.75em width0.5em depth0.25em\fi}

\begin{document}
%
\title{A nonparametric framework for inferring orders of categorical data from category-real ordered pairs}
%
%
%

\author{Chainarong~Amornbunchornvej\href{https://orcid.org/0000-0003-3131-0370}{\includegraphics[scale=.5]{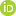}}\IEEEauthorrefmark{1},~\IEEEmembership{Member,~IEEE,}
        Navaporn~Surasvadi\href{https://orcid.org/0000-0002-8075-9839}{\includegraphics[scale=.5]{FIG/ORCIDiD_icon16x16.png}},
        Anon~Plangprasopchok\href{https://orcid.org/0000-0001-6659-580X}{\includegraphics[scale=.5]{FIG/ORCIDiD_icon16x16.png}},
        and~Suttipong~Thajchayapong\href{https://orcid.org/0000-0002-1505-4701}{\includegraphics[scale=.5]{FIG/ORCIDiD_icon16x16.png}},~\IEEEmembership{Member,~IEEE}
\thanks{All authors are at National Electronics and Computer Technology Center (NECTEC),
Pathum Thani, 12120, Thailand.}
\thanks{\IEEEauthorrefmark{1}Corresponding author (email: chainarong.amo@nectec.or.th)}
\thanks{Manuscript received November 15, 2019; revised November 15, 2019.}}

%
%

\markboth{ArXiv STAT.ME,~Vol.~X, No.~Y, November~2019}%
{C. Amornbunchornvej \MakeLowercase{\textit{et al.}}: Bare Demo of IEEEtran.cls for IEEE Journals}
%



\maketitle

\begin{abstract}
Given a dataset of careers and incomes, how large a difference of income between any pair of careers would be? Given a dataset of travel time records, how long do we need to spend more when choosing a public transportation mode $A$ instead of $B$ to travel? In this paper, we propose a framework that is able to infer orders of  categories as well as magnitudes of difference of real numbers between each pair of categories using Estimation statistics framework. Not only reporting whether an order of categories exists, but our framework also reports the magnitude of difference of each consecutive pairs of categories in the order. In large dataset, our framework is scalable well compared with the existing framework. The proposed framework has been applied to two real-world case studies: 1) ordering careers by incomes based on information of 350,000 households living in Khon Kaen province, Thailand, and 2) ordering sectors by closing prices based on 1060 companies' closing prices of NASDAQ stock markets between years 2000 and 2016. The results of careers ordering show income inequality among different careers. The stock market results illustrate dynamics of sector domination that can change over time. Our approach is able to be applied in any research area that has category-real ordered pairs. Our proposed \textit{Dominant-Distribution Network} provides a novel approach to gain new insight of analyzing category orders. The software of this framework is available for researchers or practitioners within R package: EDOIF.
\end{abstract}


%
\IEEEpeerreviewmaketitle

\input{01intro}
\input{03method}

\input{04expSet}
\input{05result}
\input{06conclusion}


%

\appendices
\input{02probformalization}

\section*{Acknowledgment}

The authors would like to thank the National Electronics and Computer Technology Center (NECTEC), Thailand, to provide our resources in order to successfully finish this work.

\ifCLASSOPTIONcaptionsoff
  \newpage
\fi



%

\bibliographystyle{IEEEtran}
\input{ms.bbl}

%


\begin{IEEEbiography}[{\includegraphics[width=1in,height=1.25in,clip,keepaspectratio]{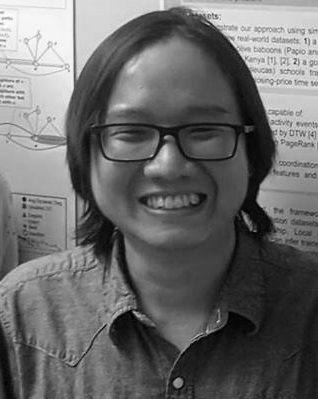}}]{Chainarong Amornbunchornvej} received the bachelor of engineering degree (with honor) in computer engineering in 2011 and the master's degree in telecommunications engineering in 2013, both from King Mongkut’s Institute of Technology Ladkrabang, Bangkok, Thailand. He received the Ph.D. degree in computer science from the University of Illinois at Chicago, IL, USA, in 2018. He is currently a researcher at National Electronics and Computer Technology Center, Thailand. He focuses on data science and statistical inference in general especially in time series analysis, causal inference, social network analysis, theoretical computer science, as well as bioinformatics.
\end{IEEEbiography}

\begin{IEEEbiography}[{\includegraphics[width=1in,height=1.25in,clip,keepaspectratio]{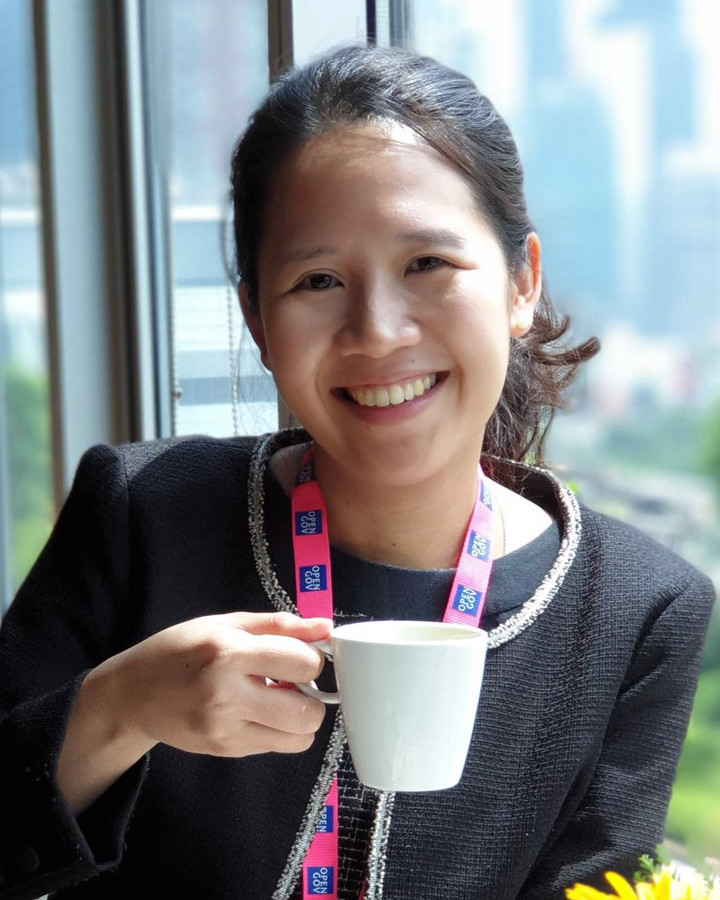}}]{Navaporn Surasvadi} is a researcher at the National Electronics and Computer Technology Center (NECTEC), Thailand. She received the BE in computer engineering (with first class honor) from Chulalongkorn University, Bangkok, Thailand and the MSc in Management Science and Engineering from Stanford University, CA, USA. She received her PhD in Operations Management from Leonard N. Stern School of Business, New York University, NY, USA in 2014. Her current research interests include data analytics and data visualization especially in strategic data for government policy planning, as well as operations management.
\end{IEEEbiography}


\begin{IEEEbiography}[{\includegraphics[width=1in,height=1.25in,clip,keepaspectratio]{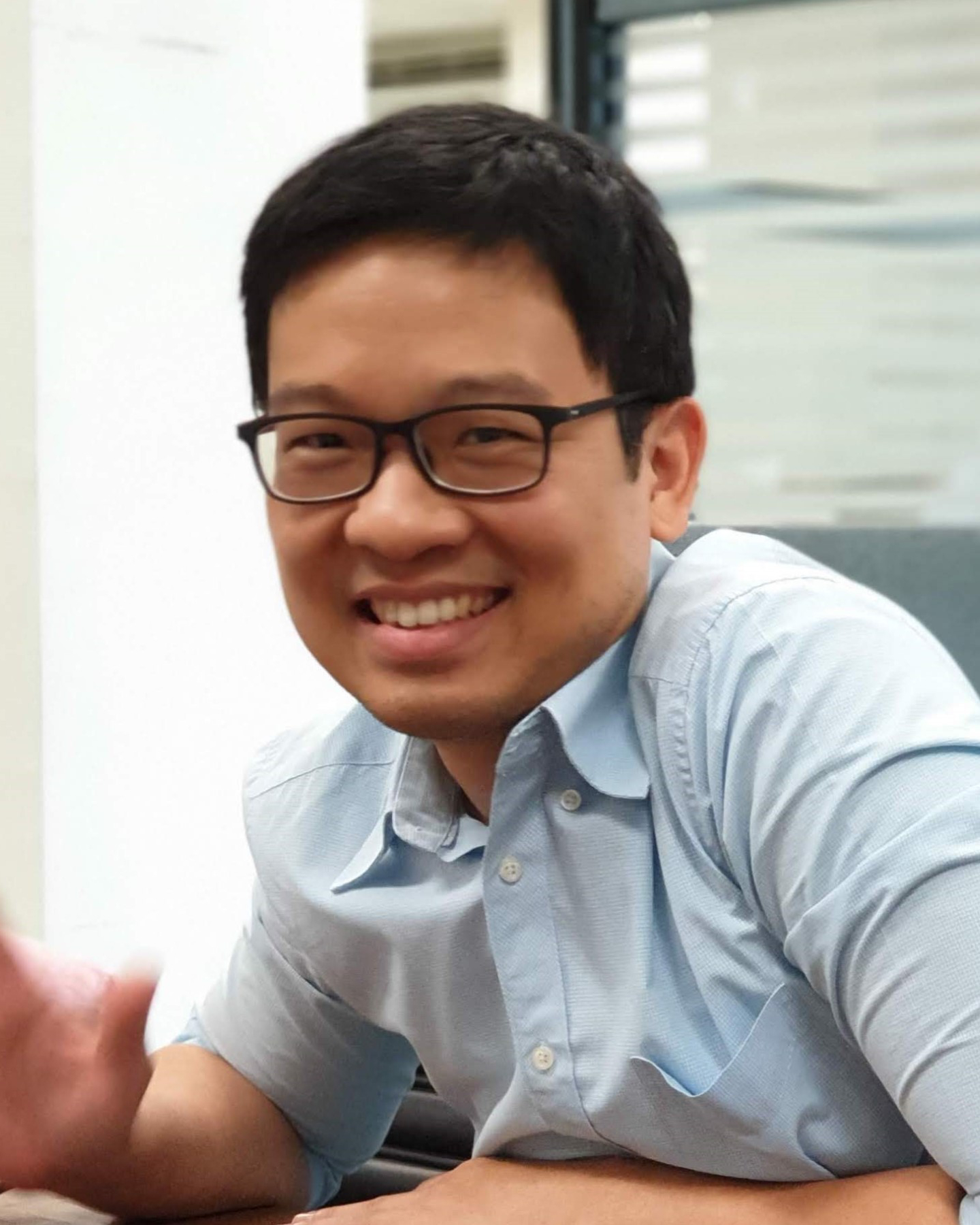}}]{Anon Plangprasopchok, Ph.D.} is a research scientist at National Electronics and Computer Technology Center (Thailand). He obtained a PhD in the Computer Science Department at the University of Southern California in 2010. His research interests lie 
in the area of applied data mining and machine learning techniques. He has been involved in several key government projects including revenue forecasting models (as a principle investigator) and data platform for poverty allievation (as a data scientist) for example.
\end{IEEEbiography}

\begin{IEEEbiography}[{\includegraphics[width=1in,height=1.25in,clip,keepaspectratio]{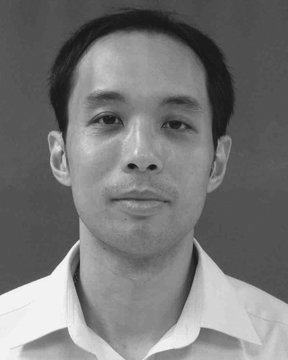}}]{Suttipong Thajchayapong}
received the M.S. and B.S. degrees in electrical and computer engineering from Carnegie Mellon University, Pittsburgh, PA, USA, and the Ph.D. degree in electrical and electronic engineering from Imperial College London, London, U.K.,He is a Researcher with National Electronics and Computer Technology Centre, National Science and Technology Development Agency, Pathumthani, Thailand. His research interests include intelligent transportation systems with emphasis on vehicular traffic monitoring and simulation, anomaly detection, and mobility and quality of service in wireless networks. Dr. Thajchayapong is a Member of ITS Thailand.
\end{IEEEbiography}



\end{document}

%% file: 01intro.tex
\section{Introduction}
%
%
%
%


\IEEEPARstart{W}{e} use an order of items with respect to their specific properties all the time to make our decision. For instance, when we plan to buy a new house, we might use an ordered list of houses based on their price or distance from a downtown. We might use travel times to order the list of transportation mode to decide which option is the best to travel from A to B, etc. 

Ordering is related to the concept of Partial order or poset~\cite{partialOrder} in Order theory. The well-known form of poset is a Directed acyclic graph (DAG) that is widely used in studying of causality~\cite{pearl2009causality,peters2017elements}, animal behavior~\cite{Amornbunchornvej:2018:CED:3234931.3201406}, social networks~\cite{kempe2003maximizing,berger2006framework}, etc. Additionally, in social science, ordering of careers based on incomes can be applied to the study of inequality in society (see Section~\ref{sec:KhonKeanCaseStudy}). 

Hence, ordering is an important concept that is used daily and  can impact society decision and scientific research. However, in the Era of Big data, inferring orders of categories items based on their real-value properties from large datasets is far from trivial.      

In this paper, we investigate the problem of inferring an order of categories based on their real-value properties, \ddoip, using poset~\cite{partialOrder} concept as well as estimating a magnitude of difference between any pair of categories. We also propose a \textit{Dominant-Distribution Network} as a representation of dominant category orders. We develop our framework based on a new concept of statistics named \textit{Estimation Statistics} principle. The aim of estimation statistics is to resolve issues of the traditional methodology, null hypothesis significance testing (NHST), that focuses on using p-value to make a dichotomous yes-no question (see Section~\ref{sec:relatedwork}).

\begin{framed}  
\noindent {\ddoip:} {In order to say that one category dominates another,  real values from one category must have higher values than other values from another category, with high probability on average (see Figure~\ref{fig:DenNoise01}). {\bf Given a set of order pairs of category-real values, the goal is to find an order list of categories with respect to their real-value distributions. If category $A$ dominates category $B$ in the list, then a probability that real-number values from $A$ is greater than an expectation of $B$'s distribution is high and not vice versa.}}
\end{framed}

In the aspect of scalability, our framework can finish analysing a dataset of 10,000 data points using 11 seconds while a candidate approach needs 300 seconds for the same dataset. The software of our proposed framework is available for researchers and practitioners with a user-friendly R package: EDOIF at~\cite{SharedLink}.

This paper is organized as follows. Section~\ref{sec:relatedwork} reviews related work, analyzing the existing gaps and how our contributions address them. Then, Section~\ref{sec:method} describes our proposed framework. Experimental setup is shown in Section~\ref{sec:expsetup} where corresponding results are discussed in Section~\ref{sec:result}. Finally, Section~\ref{sec:conclusion} concludes this paper.

\section{Related works}
\label{sec:relatedwork}
There are several NHST frameworks in both parametric (e.g. Student's t-test~\cite{10.2307/2331554}) and nonparametric (Mann-Whitney test~\cite{mann1947}) types that are able to compare two distributions and report whether one has a greater sample mean or median than another using p-value. Nevertheless, these approaches are not capable of providing a magnitude of mean difference between two distributions. Moreover, there are several issues of using only p-values to compare distributions. For instance, a null hypothesis might always get rejection since, in some system, there is always some effect in a system but an effect might be too small~\cite{cohen1995earth}. The NHST also treats distribution comparison as a dichotomous yes-no question and ignores a magnitude of difference, which might be an important information~\cite{ellis2010essential} for a research question. Besides, using only p-value information is a major issue of lack of repeatability in many research publications~\cite{halsey2015fickle}.

Hence, \textit{Estimation Statistics} has been developed as an alternative methodology to NHST. The estimation statistics is considered to be more informative than NHST~\cite{cumming2013understanding,claridge2016estimation,ho2019moving}. The primary purpose of Estimation method is to determine magnitudes of difference among distributions in terms of point estimates and  confidence intervals rather than reporting only p-value in NHST.

\begin{figure}[!ht]
\centering
\includegraphics[width=1\columnwidth]{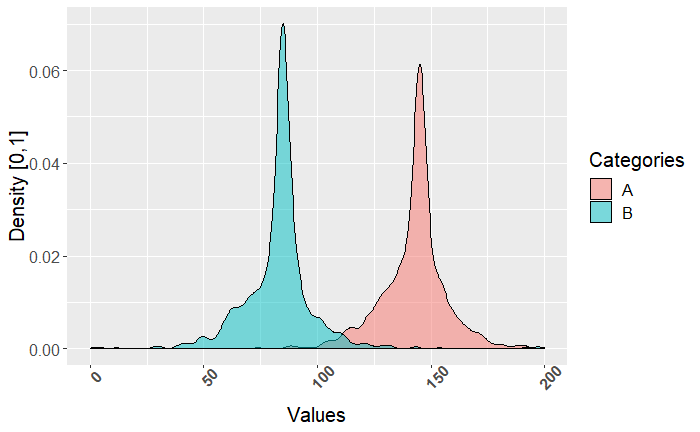}
\caption{An example of distribution of category $A$ dominates distribution of category $B$. A probability of finding a data point in $A$ that greater than $E[B]$ is greater than a probability of finding a data point in $B$ that greater than $E[A]$. }
\label{fig:DenNoise01}
\end{figure}

Recently, the Data Analysis using Bootstrap-Coupled ESTimation in R (DABESTR) framework~\cite{ho2019moving}, which is Estimation Statistics, has been developed. It mainly uses Bias-corrected and accelerated (BCa) bootstrap~\cite{doi:10.1080/01621459.1987.10478410} as a main approach to estimate a confidence interval of mean difference between distributions. BCa bootstrap is robust to a skew issue in the distribution~\cite{doi:10.1080/01621459.1987.10478410} than a percentile confidence interval and other approaches. However, it is not obvious whether BCa bootstrap is better than other approaches in the task of inferring a confidence interval of mean difference when two distributions have a high level of uniform noise (see Figure~\ref{fig:DenNoiseCom}). Moreover, DABESTR is not scalable well when there are many pairs of distributions to compare; it cannot display all confidence intervals of mean difference of all pairs in a single plot. Another issue of using BCa bootstrap is that it is too slow (see Section~\ref{sec:runningtime}) in practice compared to other approaches.  There is also no problem formalization of \ddoip, which should be considered as a problem that can be formalized by Order theory, using of partial order concept~\cite{partialOrder}.

\begin{figure*}
\centering
\includegraphics[width=1\textwidth]{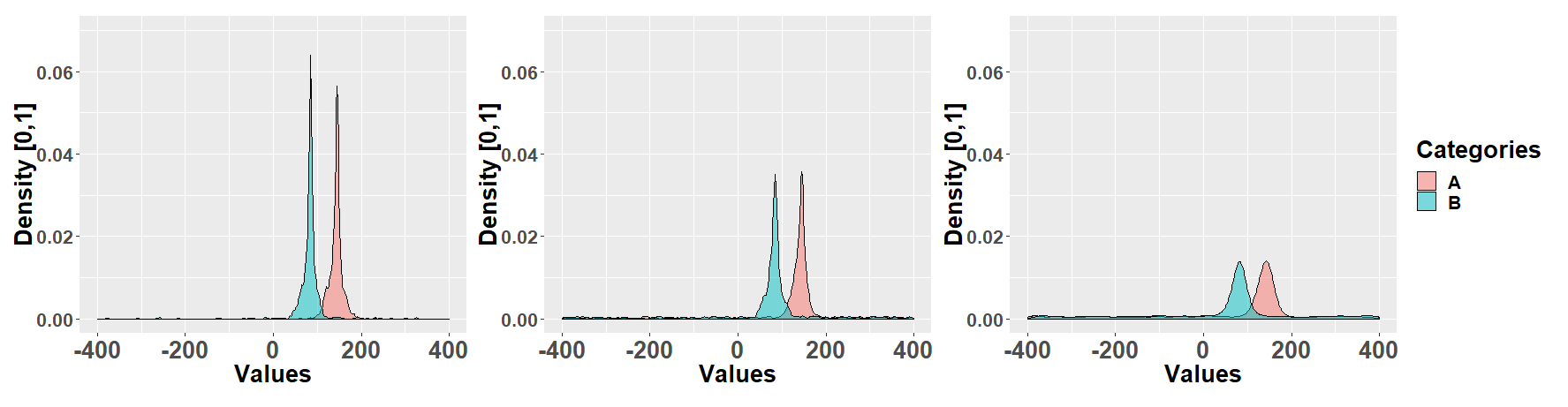}
\caption{An example of distribution of category $A$ dominates distribution of category $B$ with different degrees of uniform noise w.r.t. total data density: (left) 1\%, (middle) 20\%, and (right) 40\% of noise. The higher degree of uniform noise, the harder it is to distinguish whether $A$ dominates $B$.}
\label{fig:DenNoiseCom}
\end{figure*}

\subsection{Our Contributions}
 To fill theses gaps in the field, in this paper, we formalize \ddoip using partial order concept~\cite{partialOrder} in order theory (see Appendix~\ref{sec:probformal}). We provide a framework as a solution of \ddoip. Our framework is a non-parametric framework based on the bootstrap principle that has no assumption regarding models of data (see Appendix~\ref{sec:statInference}). We also propose to represent a dominant order with \textit{Dominant-Distribution Network} (Definition~\ref{def:DomDistNet}). Our proposed framework is capable of:
\squishlist
\item {\bf Inferring an order of multiple categories:} inferring orders of domination of categories and representing orders in the form of a graph;
\item {\bf Estimating a magnitude of difference between a pair of categories:} estimating confidence intervals of mean difference for all pairs of categories; and
\item {\bf Visualizing a network of dominant orders and magnitudes of difference among categories:} visualizing dominant orders in one graph entitled,  \textit{Dominant-Distribution Network}, as well as illustrating all magnitudes of difference of all categories pairs within a single plot that no other framework is capable of.  
\squishend

We evaluate our framework in the aspect of sensitivity analysis of uniform noise using simulation data that we posses the ground truth and compare it against several methods. To demonstrate real-world applications of our framework, we also provide two case studies.  The first is the story of inferring income orders of household careers in order to measure income inequality in Khon Kaen province, Thailand based on surveys of 350,000 households. Another case study is to use our framework to study dynamics of sector domination in NASDAQ stock market using the 1060 companies stock-closing prices between 2000 and 2016. The assessment on these two independent/irrelevant domains indicates the potential that our framework is applicable to any field of study that requires ordering of categories based on real-value data. Our \textit{Dominant-Distribution Network} (Definition~\ref{def:DomDistNet}) provides a novel approach to gain insight of analyzing category orders.

\subsection{Why confidence intervals?}
We can simply just order categories by their means or medians. However, comparing only means cannot tell us how much overlapping areas two categories have. Hence, we need mean confidence intervals to approximate the overlapping areas as well as using mean-difference confidence intervals to tell magnitude of difference between two categories. Additionally, if there are many categories and we want to infer how much pairs of categories always dominate others, then we can use a network to represent these dominant relationships. In this paper, we propose a network called a \textit{Dominant-distribution network} to represent dominant relationships among categories. 

%% file: 03method.tex
\section{Methods}
\label{sec:method}

\begin{figure*}[!ht]
\centering
\includegraphics[width=\textwidth]{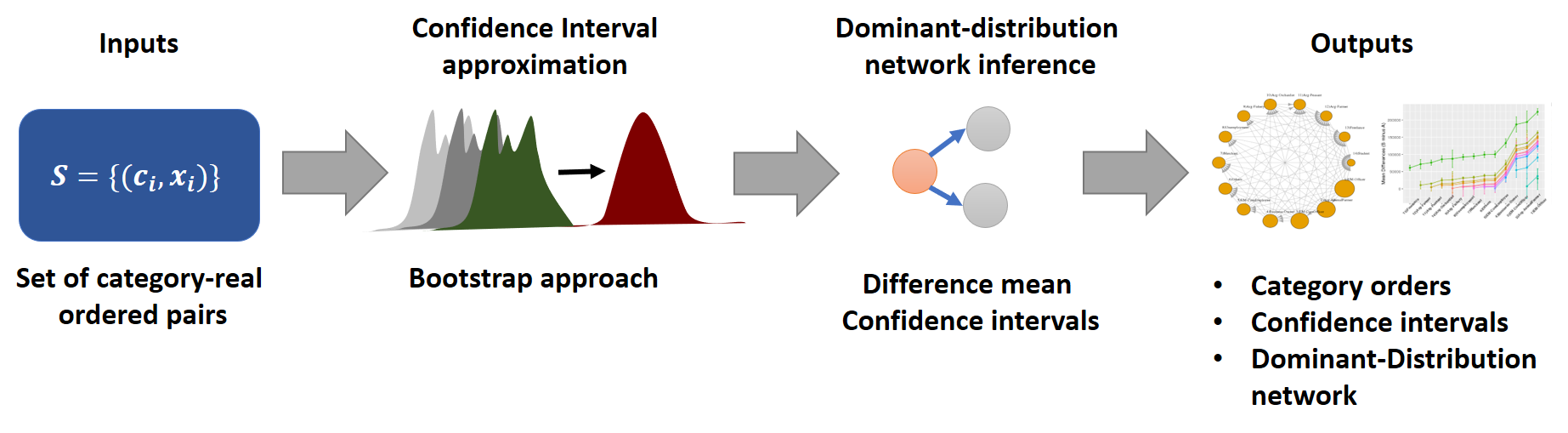}
\caption{A high-level overview of the proposed framework. }
\label{fig:FWdiagram}
\end{figure*}

For any given pair of categories $A,B$, we define an order that category $A$ dominates category $B$ using their real random variables as follows.

\begin{definition}[Dominant-distribution relation]
\label{def:EG-relation}
Given two continuous random variables $X_1 \sim \mathcal{D}_1$ and $X_2 \sim \mathcal{D}_2$ where $D_1,D_2$ are distributions. Assuming that $D_1$ and $D_2$ have the following property: $P(X_1\geq E[X_1])=P(X_2\geq E[X_2])$. We say that $D_2$ dominates $D_1$ if $P(X_1\geq E[X_2])\leq P(X_2\geq E[X_1])$; denoting $D_1 \preceq D_2$. We denote $D_1 \prec D_2$ if $P(X_1\geq E[X_2])< P(X_2\geq E[X_1])$. 
\end{definition}

Since a dominant-distribution relation is a partial order relation (Theorem~\ref{thm:prtorder}), an order always exists in any given set of ordered pairs of category and real number. For each pair of  category $A$ and $B$, we can use a bootstrap approach to infer whether $A \preceq B$ as well as using an inferred confidence interval from bootstrapping to represent a magnitude of difference between $A$ and $B$ (see Appendix~\ref{sec:statInference}).

We propose the Empirical Distribution Ordering Inference Framework (EDOIF), as a solution of \ddoip using bootstrap and additional non-parametric method. Fig~.\ref{fig:FWdiagram} illustrates an overview of our framework. Given a set of order pairs of category-real values $S=\{(c_i,x_i)\}$  as inputs of our framework where $c_i\in \mathcal{C}$ s.t. $\mathcal{C}=\{c\}$ is a set of category classes, and $x_i \in \mathbb{R}$, in this paper, we assume that for any pair $(c_i,x_i),(c_j,x_j)$ if $c_i=c_j=c$, then both $x_i$ and $x_j$ are realizations of random variables from the distribution $D'_c$.

In the first step, we infer a sample-mean confidence interval of each $D'_c$ and a mean-difference confidence interval between each pair of $D'_a$ and $D'_b$ (Section~\ref{sec:confsec}). Then, in Section~\ref{sec:DDNet}, we provide details regarding the way to infer the Dominant-distribution network.

\subsection{Confidence interval inference}
\label{sec:confsec}

\setlength{\intextsep}{0pt}
\IncMargin{1em}
\begin{algorithm2e}
\caption{ MeanBootstrapFunction}
\label{algo:MeanBootstrapFunc}
\SetKwInOut{Input}{input}\SetKwInOut{Output}{output}
\Input{ $D'=\{x_i\}$, $K$, and $\alpha$ }
\Output{ $D,CI_{\mu}$}
\begin{small}
\SetAlgoLined
\nl Setting $D= \emptyset$\;
\nl \For{ $k=1$ to $k=K$}
{
\nl    Get $D'_k$ by sampling $D'$ with replacement\;
\nl    Compute a sample mean of $D'_k$: $\bar{X}_k$\;
\nl    Add $\bar{X}_k$ to $D$\;
}
\nl Infer $(1-\alpha)100$-confidence interval of $\mu$, denoted $CI_{\mu}$, from $D$\;
\nl Return $D,CI_{\mu}$\;
\end{small}
\end{algorithm2e}\DecMargin{1em}

We separate a set $S=\{(c_i,x_i)\}$ into $D'_1,\dots,D'_C$ where $D'_c=\{x_i\}$ is a set of data point $x_i$ that has a category $c$ in $S$. We sort $D'_1,\dots,D'_C$ based on their sample means s.t. $\bar{X}_{p}\leq \bar{X}_{p+1}$ where $\bar{X}_{p},\bar{X}_{p+1}$ are sample means of $D'_p,D'_{p+1}$ respectively.

For each $D'_c$, we perform the bootstrap approach (Appendix~\ref{sec:bootstrap}) to infer the sample mean distribution $D_c$ and its $(1-\alpha)100$-confidence interval. Given $X_c\sim D_c$ and $\mu_c=E[X_c]$, the framework infers the confidence interval of $\mu_c$  w.r.t. $D_c$  denoted $CI_{\mu_c}$. Algorithm~\ref{algo:MeanBootstrapFunc} illustrates the details of how to infer $CI_{\mu_c}$ using the bootstrap approach.

\setlength{\intextsep}{0pt}
\IncMargin{1em}
\begin{algorithm2e}
\caption{ MeanDiffBootstrapFunction}
\label{algo:MeanDiffBootstrapFunc}
\SetKwInOut{Input}{input}\SetKwInOut{Output}{output}
\Input{ $D'_p,D'_q$, $K$, and $\alpha$ }
\Output{ $D_Y,CI_{\bar{Y}}$}
\begin{small}
\SetAlgoLined
\nl Setting $D_Y= \emptyset$\;
\nl \For{ $k=1$ to $k=K$}
{
\nl     Get $D'_{p,k}$ by sampling $D'_p$ with replacement\;
\nl     Get $D'_{q,k}$ by sampling $D'_q$ with replacement\;
\nl     Compute a sample means of $D'_{p,k},D'_{q,k}$: $\bar{X}_{p,k}$ and $\bar{X}_{q,k}$\;
\nl     Add the mean difference $\bar{X}_{q,k}-\bar{X}_{p,k}$ to $D_Y$\;
}
\nl Infer $(1-\alpha)100$-confidence interval of $\mu_Y$, denoted $CI_{\bar{Y}}$, from $D_Y$\;
\nl Return $D_Y,CI_{\bar{Y}}$\;
\end{small}
\end{algorithm2e}\DecMargin{1em}

In the next step, we infer an $\alpha$-mean-difference confidence interval of each pair $D'_p,D'_q$.  

Given $D_p,D_q$ are sample-mean distributions that are obtained by bootstrapping  $D'_p,D'_q$ respectively, $X_p~\sim D_p,X_q\sim D_q$, $Y=X_q-X_p$, and $\mu_Y=E[Y]$. 

The framework uses the bootstrap approach to infer sample-mean-difference distribution of $Y$ and the  $(1-\alpha)100$-confidence interval of $\mu_Y$. Algorithm~\ref{algo:MeanDiffBootstrapFunc} illustrates the details of how to infer $CI_{\bar{Y}}$ using the bootstrap approach in general.

Even though we can use a normal confidence interval as a confidence interval in line 6 of Algorithm~\ref{algo:MeanBootstrapFunc} and line 7 of Algorithm~\ref{algo:MeanDiffBootstrapFunc} (see Lemma~\ref{lemma:confintv}), the normal bound has an issue when a distribution is skew~\cite{ho2019moving,doi:10.1080/01621459.1987.10478410}. Hence, we deploy both percentile confidence intervals and Bias-corrected and accelerated (BCa) bootstrap~\cite{doi:10.1080/01621459.1987.10478410} to infer both confidence intervals: $CI_{\mu_c}$ and $CI_{\bar{Y}}$. \\

For a percentile confidence interval inference (our default option) and BCa bootstrap, we deploy a standard library of bootstrap approaches in R ``boot" package~\cite{r2011r,davison1997bootstrap,Rboot1}.

\subsection{Dominant-distribution network inference}
\label{sec:DDNet}

The first step of inferring Dominant-distribution network $G=(V,E)$ in Definition~\ref{def:DomDistNet} is to infer whether $D_p\preceq_{\alpha} D_q$. 

In the network $G=(V,E)$, a node $p \in V$ represents $D_p$ and $(q,p)\in E$ if $D_p \preceq_\alpha D_q$.

Given $X_p~\sim D_p,X_q\sim D_q$, $Y=X_q-X_p$, we can check the normal lower bound of $CI_{\bar{Y}}$ in Lemma~\ref{lemma:confintv} that we mentioned in Section~\ref{sec:DDrelaionInfer}. If the lower bound $\bar{Y} - z_{\frac{\alpha}{2} } \frac{s_{Y} }{\sqrt{k}}$ is greater than zero, then $D_p\preceq_{\alpha} D_q$. However, we deploy Mann-Whitney test~\cite{mann1947} to infer whether $D_p\preceq_{\alpha} D_q$ due to its robustness (see the Result Section). Along with Mann-Whitney test~\cite{mann1947}, we also deploy p-value adjustment method by Benjamini and Yekutieli (2001)~\cite{benjamini2001control} to reduce the false positive issue.

In the next step, for each $D_p$, we add node $p$ to $V$. For any pair $D_p,D_q$, if $D_p \preceq_{\alpha} D_q$, then $(q,p) \in E$. One of the properties we have for $G$ is that the set of nodes that are reachable by the path from $q$ is a set of distributions of which $D_q$ dominates them. 

\subsection{Visualization}

We use ggplots package~\cite{wickham2016ggplot2} to create mean confidence intervals (e.g. Figure~\ref{fig:KhonkeanCIraw}) and mean-difference confidence intervals (e.g. Figure~\ref{fig:KhonKeanDiffMeans}) plots. For a dominant-distribution network, we visualize it using iGraph package~\cite{csardi2006igraph} (e.g. Figure~\ref{fig:KhonkeanDNet}).

%% file: 04expSet.tex
\section{Experimental setup}
\label{sec:expsetup}
We use both simulation and real-world datasets to evaluate our method performance. 
\subsection{Simulation data for sensitivity analysis}
\label{sec:simdata}
We simulated datasets from mixture distributions, which consists of a normal distribution, Cauchy distribution, and uniform distribution.  The random variable $X$ is defined as follows.

\begin{equation}
  X\sim\left\{
  \begin{array}{@{}ll@{}}
	\mathcal{N}(\mu_0,\sigma_0), & \text{with probability } 0.5 \\
   \mathcal{C}(x_0,\gamma), & \text{with probability } (0.5 - p_1)\\
   \mathcal{U}(L_1,U_1), & \text{with probability } p_1 \\
  \end{array}\right.
\label{eq:d1}
\end{equation}

Where $\mathcal{N}(\mu_0,\sigma_0)$ is a normal distribution with mean $\mu_0$ and variance $\sigma_0^2$, $\mathcal{C}(x_0,\gamma)$ is a Cauchy distribution with location $x_0$ and scale $\gamma$, $\mathcal{U}(L_1,U_1)$ is a uniform distribution with the minimum number $L_1$ and maximum number $U_1$, and $p_1$ is a value that represents a level of uniform noise. When the $p_1$ increases, the ratio of uniform distribution in the mixture distribution increases.   We set $p_1=\{0.01,0.05,0.10,0.15,0.20,0.25,0.30,0.35,0.40\}$ to generate simulation datasets in order to perform the sensitivity analysis. 

\begin{figure}[!ht]
\centering
\includegraphics[width=.4\columnwidth]{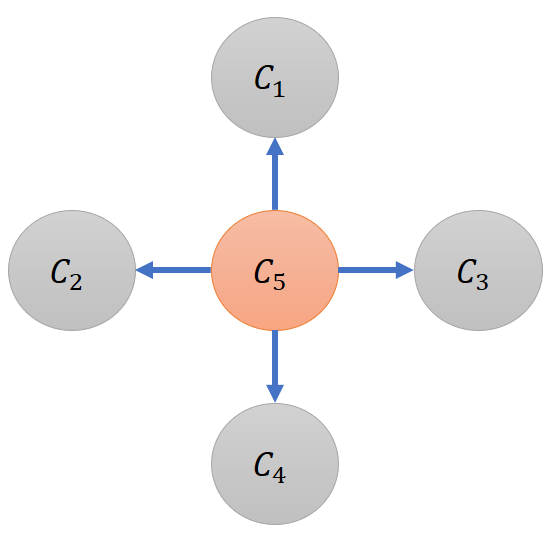}
\caption{A dominant-distribution network $G$ of simulation datasets}
\label{fig:Simgraph}
\end{figure}

In all simulation datasets, there are five categories: $C_1,\dots,C_5$. The dominant-distribution relations of these categories are represented as a dominant-distribution network $G$. The network $G$ is shown in Figure~\ref{fig:Simgraph}. Only $C_5$ dominates others.
In this paper, for $C_1,\dots,C_4$, we set $\mu_0=80,\sigma_0=16,x_0=85,\gamma=2,L_1=-400,U_1=400$ to generate realizations of $X$. For $C_5$,  we set $\mu_0=140,\sigma_0=16,x_0=145,\gamma=2,L_1=-400,U_1=400$.

Because uniform distribution in the mixture distribution has the range between -400 and 400, but all areas of distributions of $C_1,\dots,C_5$ are within $[-400,400]$, a method has more issue to distinguish whether $C_i \preceq C_j$ for any  $C_i,C_j \in\{C_1,\dots,C_5\}$ when we increase $p_1$ (see Fig~\ref{fig:DenNoiseCom}).

The main task of inference here is to measure whether a given method can infer that $C_i \preceq C_j$ w.r.t. a network in Figure~\ref{fig:Simgraph} from these simulation datasets. We generate 100 datasets for each different value of $p_1$. In total, there are 900 datasets. 

To measure the performance of ordering inference, we define true positive (TP), false positive (FP), and false negative (FN) in order to calculate precision, recall, and F1 score as follows. Given any pair of categories $C_i,C_j$, TP is when both ground truth (Figure~\ref{fig:Simgraph}) and inferred result agree that $C_i \preceq C_j$ is true. FP is  when a method infers that $C_i \preceq C_j$ but the ground truth disagrees. FN is when the ground truth has $C_i \preceq C_j$ but an inferred result from the method disagrees. 

In the task of inferring whether $C_i \preceq C_j$, we compared our approach (Mann-Whitney test~\cite{mann1947} with p-value adjustment method~\cite{benjamini2001control}) against 1) t-test with Pooled Standard Deviation~\cite{cohen1998statistical}, 2) t-test with p-value adjustment~\cite{benjamini2001control}, 3) BCa bootstrap, and 4) percentile bootstrap. For both BCa bootstrap, and percentile bootstrap, we decide whether $C_i \preceq C_j$ based on the lower bound of confidence intervals of mean difference between $C_i$ and $C_j$. If the lower bound is positive, then $C_i \preceq C_j$, otherwise, $C_i \not\preceq C_j$. 

\subsection{Real-world data: Thailand's population household information}
\label{sec:realdata1}
This dataset was obtained from Thailand household-population surveys from Thai government in 2018~\cite{amornbunchornvej2019identifying}. The purpose of this survey was to analyze the Multidimensional Poverty Index (MPI)~\cite{alkire2010multidimensional,alkire2018global}, which is considered as a current main poverty index that the United Nations (UN) uses. We deployed the data of household incomes and careers information from 353,910 households of Khon Kaen province, Thailand to perform our analysis. We categorized careers of heads of household into 14 types: student (student), freelance (Freelance), plant farmer (AG-Farmer), peasant (AG-Peasant), orchardist (AG-Orchardist), fishery (AG-Fishery), animal farmer (AG-AnimalFarmer), unemployment (Unemployment), merchant (Merchant), company employee (EM-ComEmployee),  business owner (Business-Owner), government's company employee (EM-ComOfficer), government officer (EM-Officer), and others (Others). The incomes in this dataset are annual incomes of households and the unit of incomes is in Thai Baht (THB). 

Given a set of ordered pairs of career and household income, we analyzed the income gaps of different types of careers in order to study the inequality of population w.r.t. people careers.

\subsection{Real-world data: NASDAQ Stock closing prices}
\label{sec:realdata2}
This NASDAQ stock-market dataset has been obtained by the work in~\cite{Amornbunchornvej:2018:CED:3234931.3201406} from Yahoo! Finance.\footnote{http://finance.yahoo.com/} The dataset was collected from January 2000 to January 2016. It consist of a set of time series of stock closing prices of 1060 companies. Each company time series has a total length as 4169 time-steps. Due to the high variety of company sectors, in this study, we separated these time series into five sectors: `Service \& Life Style', `Materials', `Computer', `Finance', and `Industry \& Technology'.  

In order to observe the dynamics of domination, we separated time series into two intervals: 2000-2014, and 2015-2016. For each intervals, we aggregated the entire time series using median. 

Given a set of ordered pairs of closing-price median and sector, the purpose of this study is to find which sectors dominated others in each interval.

\subsection{Parameter settings}
\label{sec:para}

We set the significant level $\alpha =0.05$ and the number of times of sampling with replacement for a bootstrap approach is $1000$ for all experiments unless stated otherwise.

\subsection{Running time}
\label{sec:runningtime}
\begin{figure}[!ht]
\centering
\includegraphics[width=1\columnwidth]{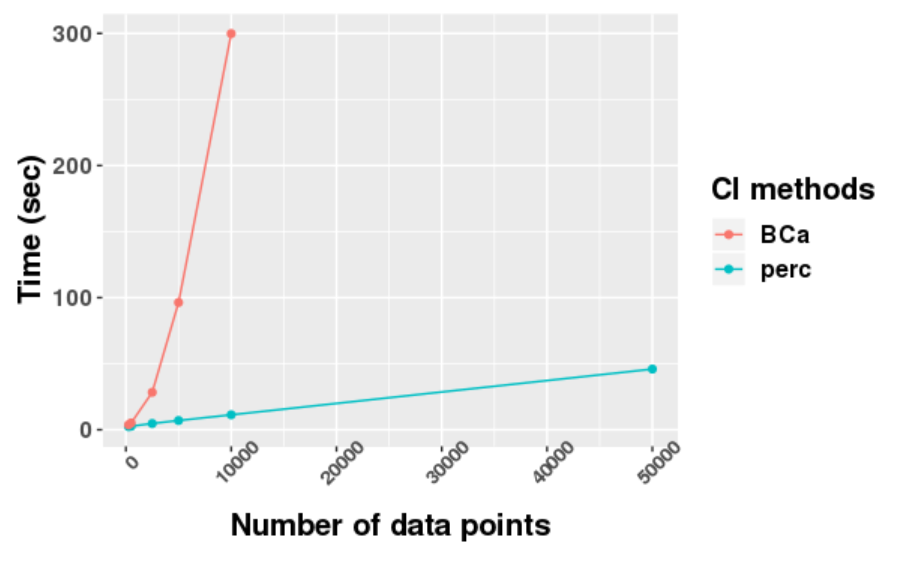}
\caption{A comparison of running time between two methods of Bootstrap confidence intervals.}
\label{fig:TimeCost}
\end{figure}

In this experiment, we compared the running time of two methods of bootstrapping to infer confidence intervals: BCa bootstrap (BCa) and percentile (perc) approaches using simulation datasets from the previous section.\footnote{The computer specification that we used in this experiment is Dell 730, with CPU Intel Xeon E5-2630 2.4GHz, and Ram 128 GB.} We set the number of times of bootstrapping as 4000 rounds. In Figure~\ref{fig:TimeCost}, the result is shown that BCa method was a lot slower than the percentile approach. In the dataset of 10,000 data points, the BCa bootstrap required the running time around 300 seconds while the percentile approach required only 11 seconds. Besides, for a dataset that has 500,000 data points, percentile approach was able to finish running around 11 minutes. This indicates that the percentile approach is scalable better than BCa bootstrap. Hence, for a large dataset, we recommend users to use the percentile approach since it is fast and the performance is comparable or even better than BCa method that we will show in the next section.

%% file: 05result.tex
\section{Results}
\label{sec:result}
\subsection{Simulation results}
In this section, we report the results of our analysis from simulation datasets (Section~\ref{sec:simdata}). The main task is the ordering inference; determining whether $A \preceq B$ for all pairs of categories.

\begin{table}
\caption{The categories ordering inference result; each approach is used to infer orders of any pair of two categories w.r.t. the real-values within each category.}
\label{tab:CatOrInferRes}
\begin{tabular}{c|c|c|c|}
\cline{2-4}
                                      & Precision & Recall & F1 scores \\ \hline
\multicolumn{1}{|c|}{ttest (pool.sd)} & 0.61      & 0.52   & 0.55      \\ \hline
\multicolumn{1}{|c|}{ttest}           & 0.72      & 0.72   & 0.72      \\ \hline
\multicolumn{1}{|c|}{Bootstrap: BCa}  & 0.70      & 0.67   & 0.68      \\ \hline
\multicolumn{1}{|c|}{Bootstrap: Perc} & 0.73      & 0.68   & 0.70      \\ \hline
\multicolumn{1}{|c|}{EDOIF (Mann-Whitney)}    & 0.77      & 0.85   & 0.81      \\ \hline
\end{tabular}
\end{table}

Table~\ref{tab:CatOrInferRes} illustrates the categories ordering inference result. Each value in the table is the aggregate results of datasets from different values of $p_1$: $p_1=\{0.01,0.05,0.10,0.15,0.20,0.25,0.30,0.35,0.40\}$. The table shows that our approach (using Mann-Whitney) performance is above all approaches.  While ttest (pool.sd) performed the worst, the traditional t-test performed slightly better than both bootstrap approaches. Comparing between BCa and percentile bootstraps, the performance of percentile bootstrap is slightly better than BCa bootstrap. Even though BCa bootstrap covers the skew issue better than percentile bootstrap~\cite{ho2019moving,doi:10.1080/01621459.1987.10478410}, our result indicates that percentile bootstrap is more accurate than BCa when the noise presents in the task of ordering inference.

\begin{figure}
    \centering
    \includegraphics[width=.8\columnwidth]{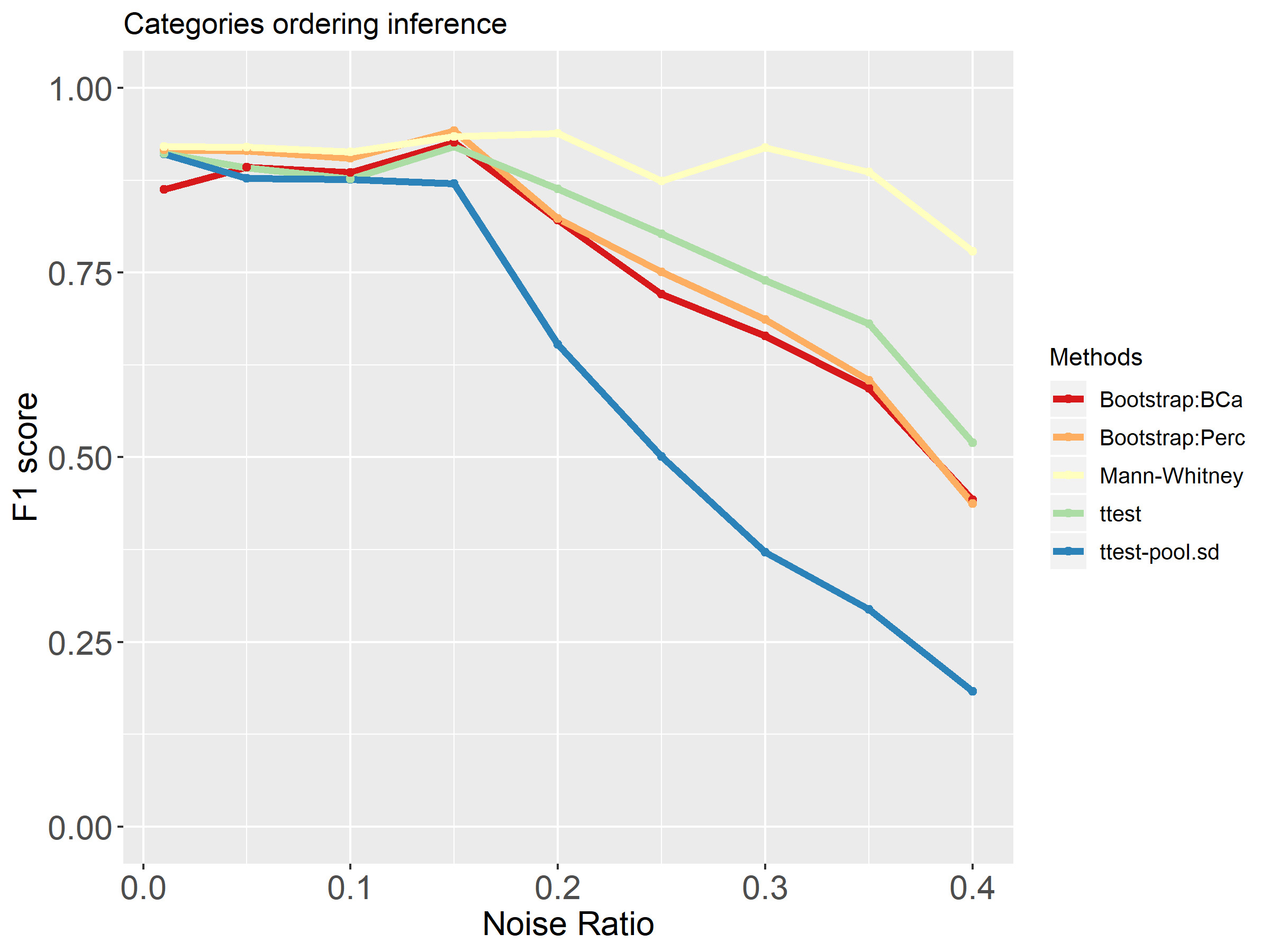}
    \caption{The sensitivity analysis of categories ordering inference. The simulation datasets containing different levels of noise were deployed for the experiment (best viewed in colour codes). }
    \label{fig:F1OrderInferRes}
\end{figure}

Figure~\ref{fig:F1OrderInferRes} shows the result of sensitivity analysis of all approaches when the uniform noise presents in different degrees. The horizontal axis represents noise ratios and the vertical axis represents F1 score in the task of ordering inference. According to Figure~\ref{fig:F1OrderInferRes}, our approach (using Mann-Whitney) performed better than all methods in all levels of noise. t-test preformed slightly better than both bootstraps approaches.  Both bootstrap methods performance are quite similar. The t-test with (pool.sd) performed the worst. Both Table~\ref{tab:CatOrInferRes} and Figure~\ref{fig:F1OrderInferRes} illustrate the robustness of our approach.

\subsection{Case study: Ordering career categories based on Thailand's household incomes  in Khon Kaen province}
\label{sec:KhonKeanCaseStudy}

\begin{figure}[!ht]
\centering
\includegraphics[width=1\columnwidth]{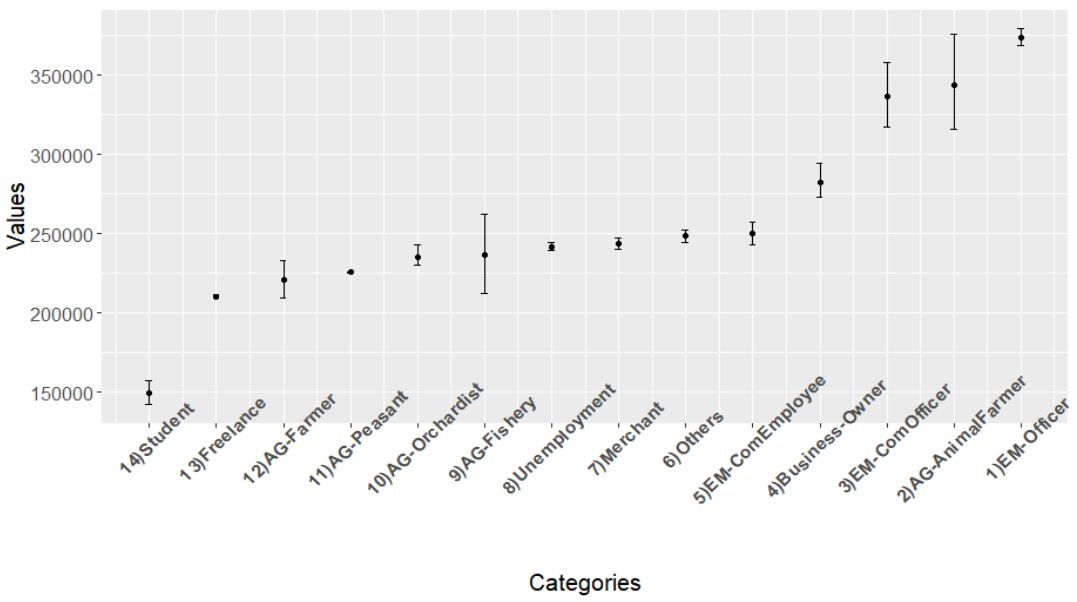}
\caption{Confidence intervals of household incomes of the population from Khon Kaen province categorized by careers. }
\label{fig:KhonkeanCIraw}
\end{figure}

In this section, we report the orders of careers based on incomes of population in Khon Kaen province, Thailand. Due to the expensive cost of computation of BCa bootstrap, in this dataset, since there are 353,910 data points, we used percentile bootstrap as a main method. Figure~\ref{fig:KhonkeanCIraw} illustrates the bootstrap-percentile confidence intervals of mean incomes of all careers with an order. 

A government officer (EM-Officer) class is ranked as the 1st place of career that has the highest mean income, while a student class has the lowest mean income. 

\begin{figure}[!ht]
\centering
\includegraphics[width=1\columnwidth]{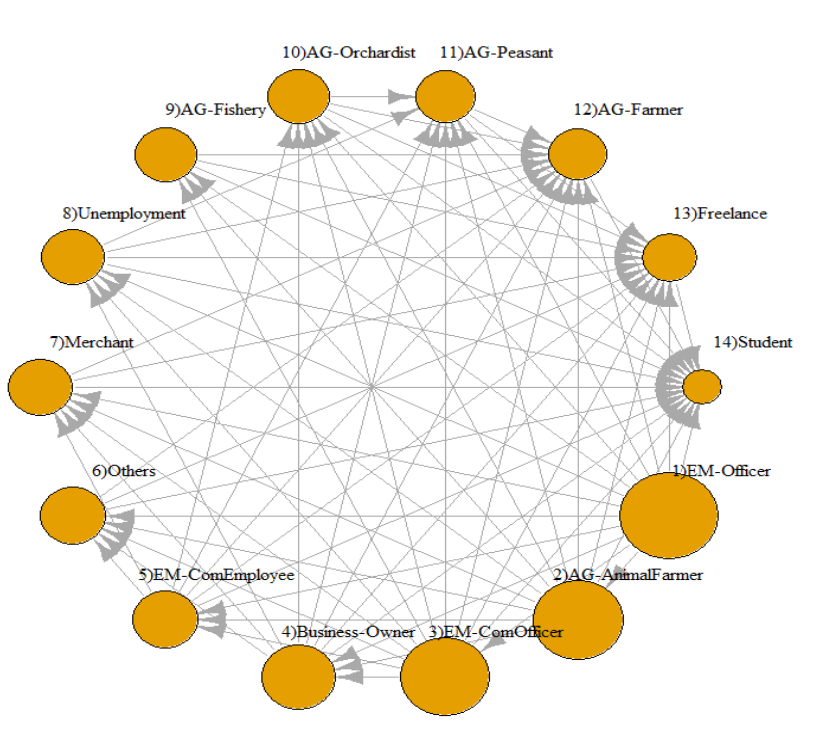}
\caption{A dominant-distribution network of household incomes of the population from Khon Kaen province categorized by careers.   A node size represents a magnitude of sample mean of incomes in a career node. }
\label{fig:KhonkeanDNet}
\end{figure}

Figure~\ref{fig:KhonkeanDNet} shows orders of dominant-distribution relations of career classes in a form of a dominant-distribution network. It shows that a government officer (EM-Officer) class dominates all career classes. In a dominant-distribution network, its network density represents a level of domination; higher network density implies there are many categories that are dominated by others. The network density of the network is 0.79. Since the network density is high, a higher-rank career class seems to dominate a lower-rank career class with high probability. This implies that different careers provide different incomes. In other words, gaps between careers are high. Figure~\ref{fig:KhonKeanDiffMeans} provides the magnitudes of income-mean difference between pairs of careers in the form of confidence intervals. It shows us that the majority of pairs of different careers have gaps of annual incomes at least 25,000 THB (around \$800 USD)!   

\begin{figure*}[!ht]
\centering
\includegraphics[width=2\columnwidth]{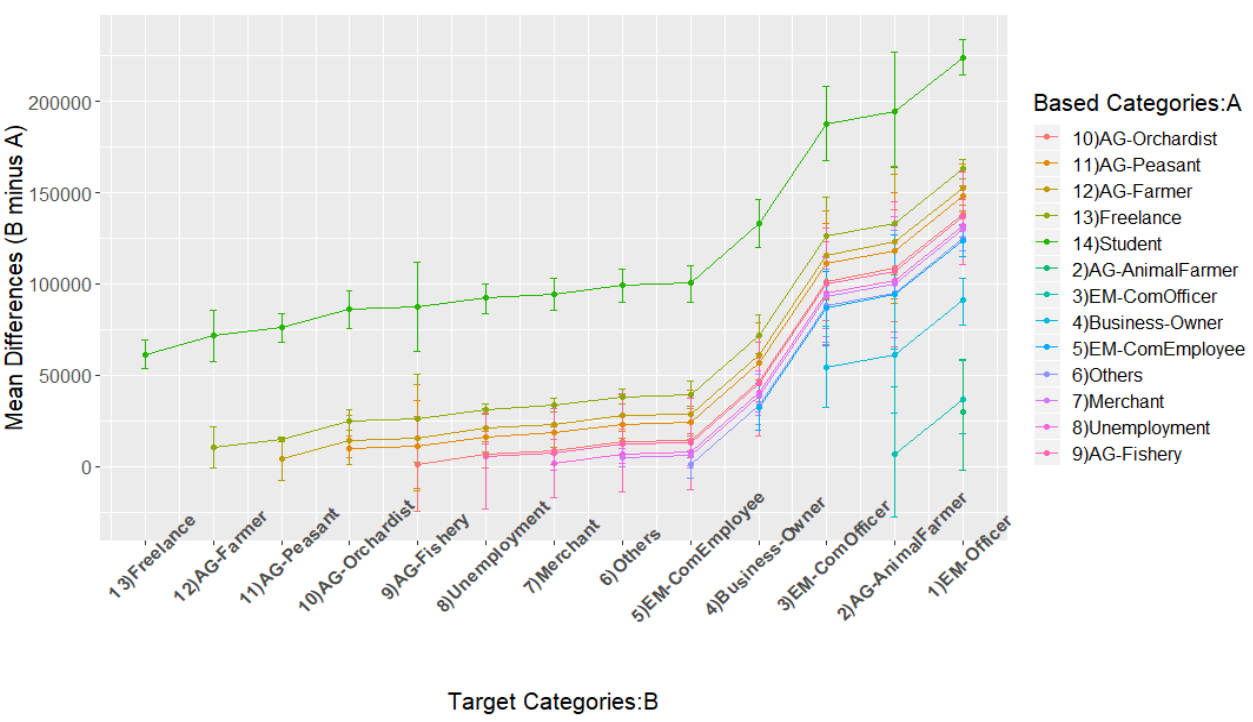}
\caption{Mean-difference confidence intervals of different careers based on household incomes of the population from Khon Kaen province categorized by careers. }
\label{fig:KhonKeanDiffMeans}
\end{figure*}

Since one of definitions of economic inequality is income inequality~\cite{kuznets1955economic,kawachi1999income,oishi2011income}, there is a high degree of career-income inequality in this area. In societies with a more equal
distribution of incomes, people are healthier ~\cite{kawachi1999income}. This inequality might lead to other issues such as health issue. Moreover, the income inequality is associate with happiness of people~\cite{oishi2011income}.  This case study shows that using our dominant-distribution network and mean-difference confidence intervals is a novel way of studying career-income inequality.

\subsection{Case study: Ordering aggregate-closing prices of NASDAQ stock market based on sectors}

This case study reveals the dynamics of sector domination in NASDAQ stock market.  We report the patterns of dominate sectors that change over time in the market. 

\begin{figure*}[!ht]
\centering
\includegraphics[width=2\columnwidth]{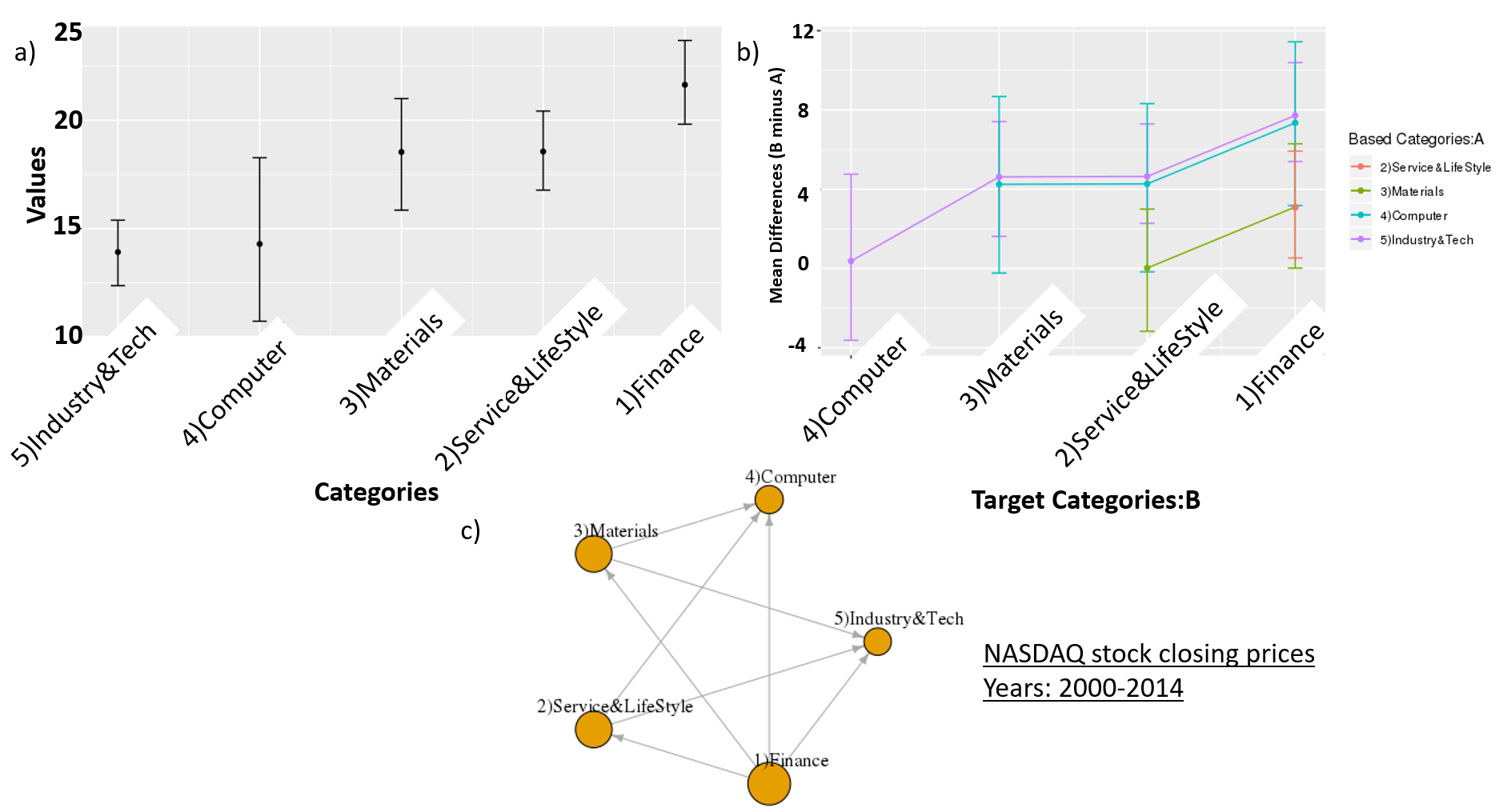}
\caption{The sectors ordering result of NASDAQ stock closing prices from 1060 companies between 2000 and 2014.    a) Confidence intervals of closing prices of sectors. b)  Confidence intervals of difference means of closing prices among sectors. c) A dominant-distribution network of sectors. }
\label{fig:NDres00-14}
\end{figure*}

Figure~\ref{fig:NDres00-14} shows the sectors ordering result of NASDAQ stock closing prices from 1060 companies between 2000 and 2014. The dominated sector is `Finance' sector that dominates all other sectors. Due to the high network density of the dominant-distribution network at 0.8, there are large gaps between sectors in this time interval.

\begin{figure*}[!ht]
\centering
\includegraphics[width=2\columnwidth]{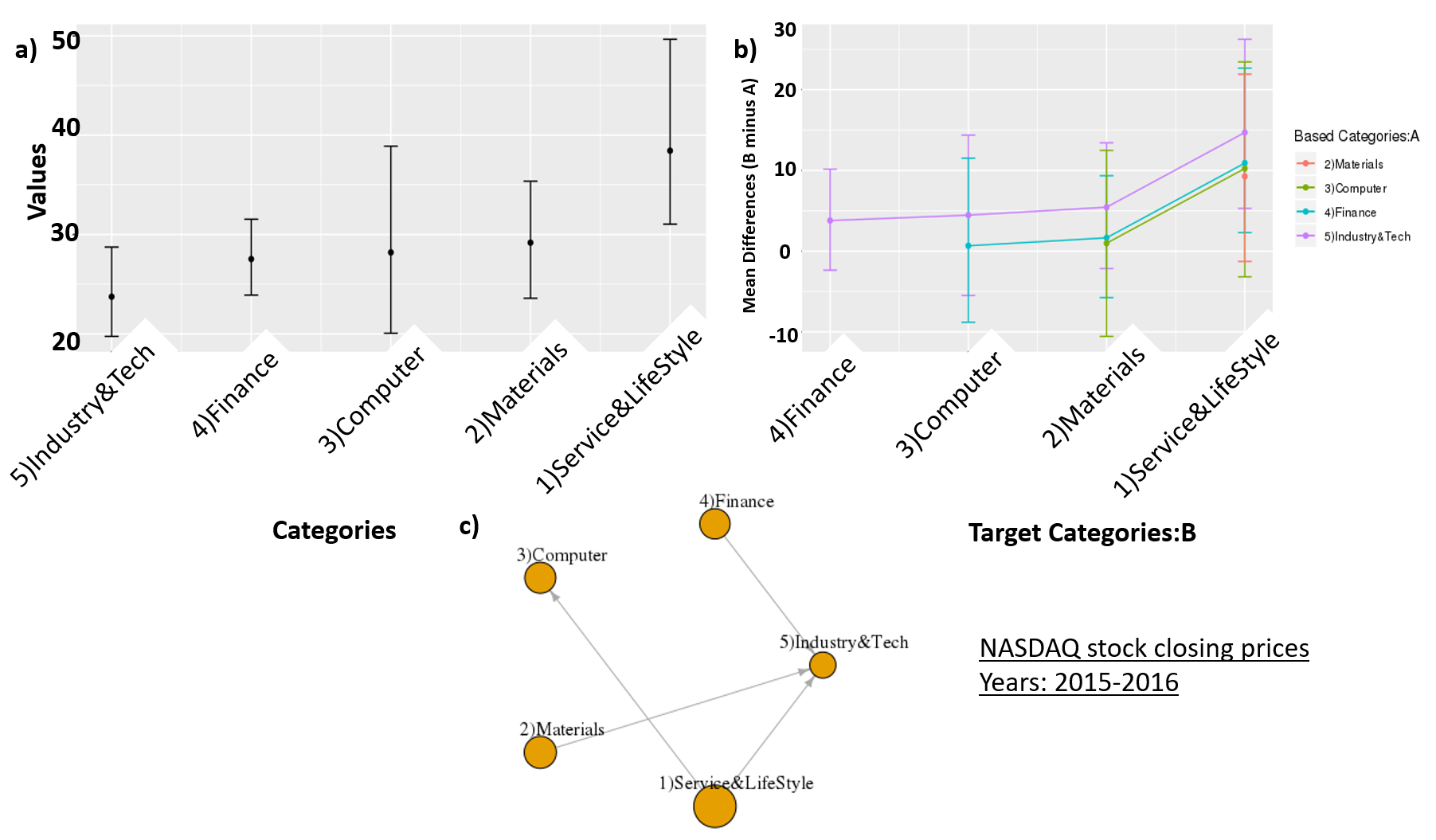}
\caption{The sectors ordering result of NASDAQ stock closing prices from 1060 companies between 2015 and 2016. We separated companies into five main sectors: `Service \& Life Style', `Materials', `Computer', `Finance', and `Industry \& Technology'. a) Confidence intervals of closing prices of sectors. b)  Confidence intervals of difference means of closing prices among sectors. c) A dominant-distribution network of sectors. }
\label{fig:NDres15-16}
\end{figure*}

On the other hand, in Figure~\ref{fig:NDres15-16}, the sectors result ordering of NASDAQ stock between 2015 and 2016 demonstrates that there is no sector that dominate all other sectors. The network density is 0.4, which implies that the level of domination is less than the previous interval. The Finance sector is ranked as 4th position in the order.  It is not because the Finance sector has a lower closing price in recent years, but all other sectors have higher closing prices lately. The computer sector has a higher closing price lately compared to the previous time interval, which is consistent with the current situation that the IT development (e.g. big data analytics, AI, block chain) impacts many business scopes significantly~\cite{du2018current}.  

%% file: 06conclusion.tex
\section{Conclusion}
\label{sec:conclusion}
In this paper, we proposed a framework that is able to infer orders of  categories based on their expectation of real-number values using Estimation statistics framework. Not only reporting whether an order of categories exists, but our framework also reports the magnitude of difference of each consecutive pairs of categories in the order using confidence intervals and a dominant-distribution network. In large dataset, our framework is scalable well using percentile bootstrap approach compared with the existing framework: DABESTR that uses BCa bootstrap. The proposed framework was applied to two real-world case studies: 1) ordering careers by 350,000 household incomes from the population of Khon Kaen province in Thailand, and 2) ordering sectors based on 1060 companies' closing prices of NASDAQ stock markets between years 2000 and 2016. The results of careers ordering showed income-inequality among different careers in a dominant-distribution network. The stock market results illustrated dynamics of sectors that dominate the market can be changed over time. The encouraging results show that our approach is able to be applied to any other research area that has category-real ordered pairs. Our proposed \textit{Dominant-Distribution Network} provides a novel approach to gain new insight of analyzing category orders. The software of this framework is available for researchers or practitioners with a  user-friendly R package at \cite{SharedLink}.

%% file: 02probformalization.tex
\section{Problem formalization}
\label{sec:probformal}
In this section, we provide the details regarding that a dominant-distribution relation is a partial order as well as providing the problem formalization of \ddoip. In the first step, we provide the concept of equivalent distributions.

\begin{proposition}
\label{pros:EquivR}
Let $D_1,D_2$ be distributions such that $D_1 \preceq D_2$ and $D_2 \preceq D_1$, then $D_1,D_2$ are equivalent distributions denoted $D_1\equiv D_2$.
\end{proposition}
\begin{proof}

When $D_1 \preceq D_2$ and $D_2 \preceq D_1$, the first obvious case is $P(X_1\geq E[X_2])= P(X_2\geq E[X_1])$. For the case that  $D_1 \prec D_2$ and $D_2 \prec D_1$, this cannot happen because of contradiction. Hence,  $D_1 \preceq D_2$ and $D_2 \preceq D_1$ implies only $P(X_1\geq E[X_2])= P(X_2\geq E[X_1])$.   
\end{proof}

We provide a relationship between expectations of distribution and a dominant-distribution relation below.

\begin{proposition}
\label{pros:DomExp}
Let $D_1,D_2$ be distributions, and $X_1\sim D_1, X_2 \sim D_2$ s.t. $P(X_1\geq E[X_1])=P(X_2\geq E[X_2])$.  $E[X_1]\leq E[X_2]$ if and only if $D_1 \preceq D_2$.
\end{proposition}
\begin{proof}
In the forward direction, suppose $E[X_1]\leq E[X_2]$. Because  the center of $D_2$ is on the right of $D_1$ in the real-number axis,  hence, $P(X_2\geq E[X_1])$ covers almost areas of $D_2$ distribution except the area of $P(X_2<E[X_1])$. In contrast, $P(X_1\geq E[X_2])$ covers only a tiny area in the far right of $D_1$. This implies that  $P(X_1\geq E[X_2])\leq P(X_2\geq E[X_1])$ or $D_1 \preceq D_2$. 

In the backward direction, we use the proof by contradiction. Suppose $D_1 \preceq D_2$. Because $D_1 \preceq D_2$ implies $P(X_1\geq E[X_2])\leq P(X_2\geq E[X_1])$ and $P(X_1\geq E[X_1])=P(X_2\geq E[X_2])$, then we have the following implications.

Let assume that $E[X_2] < E[X_1]$.
This implies that $P(X_1\geq E[X_1])< P(X_1\geq E[X_2])$. Since $P(X_1\geq E[X_1])=P(X_2\geq E[X_2])$, we have

\begin{equation}
\label{eq:prf1}
    P(X_2\geq E[X_2])< P(X_1\geq E[X_2]).
\end{equation}
Assuming $E[X_2] < E[X_1]$, we also have

\begin{equation}
\label{eq:prf2}
    P(X_2\geq E[X_1])< P(X_2\geq E[X_2]).
\end{equation}
By combining inequation~\ref{eq:prf1} and inequation~\ref{eq:prf2}, we have

\begin{equation}
\label{eq:prf3}
    P(X_2\geq E[X_1])< P(X_1\geq E[X_2]).
\end{equation}
The inequation~\ref{eq:prf3} contradicts with the requirement of $D_1 \preceq D_2$, which is  $P(X_1\geq E[X_2])\leq P(X_2\geq E[X_1])$! 
Therefore, $E[X_1]\leq E[X_2]$.
\end{proof}

In the next step, we show that a dominant-distribution relation has a transitivity property.

\begin{proposition}
\label{pros:trns}
Let $D_1,D_2,D_3$ be distributions such that $D_1 \preceq D_2$, $D_2 \preceq D_3$, then $D_1\preceq D_3$.
\end{proposition}
\begin{proof}
According to Proposition~\ref{pros:DomExp}, $D_1 \preceq D_2$ implies  $E[X_1]\leq E[X_2]$. 

Now, we have $E[X_1]\leq E[X_2]\leq E[X_3]$. The $D_3$ distribution must be on the right hand side of $D_1$. Hence, $P(X_1\geq E[X_3])\leq P(X_3\geq E[X_1])$, which implies $D_1\preceq D_3$.  
\end{proof}

Now, we are ready to conclude that a dominant-distribution relation is a partial order.

\begin{theorem}
\label{thm:prtorder}
Given a set $S$ of continuous distributions s.t. for any pair $D_1,D_2 \in S$, given $X_1\sim D_1,X_2 \sim D_2$, $P(X_1\geq E[X_1])=P(X_2\geq E[X_2])$. The {\sc dominant-distribution relation} is a partial order over a set $S$~\cite{partialOrder}.
\end{theorem}
\begin{proof} A relation is a partial order over a set $S$ if it has the following properties: Antisymmetry, Transitivity, and Reflexivity.
\squishlist
\item {\bf Antisymmetry:} if $D_1\preceq D_2$ and $D_2\preceq D_1$, then $D_1 \equiv D_2$ by Proposition~\ref{pros:EquivR}.
\item {\bf Transitivity:} if $D_1 \preceq D_2$, $D_2 \preceq D_3$, then $D_1\preceq D_3$ by Proposition~\ref{pros:trns}.
\item {\bf Reflexivity:} $\forall D, \; D\preceq D$.
\squishend
Therefore, by definition, the {\sc dominant-distribution relation} is a partial order over a set of continuous distributions.
\end{proof}

Suppose we have $D_1 \preceq D_2$ and $X_1\sim D_1, X_2\sim D_2$. We can have $Y=X_2-X_1$ as a random variable that represents the magnitude of difference between two distributions. Suppose $\mu_Y$ is the true mean of $Y$'s distribution, our next goal is to find the confidence interval of $\mu_Y$.

\begin{definition}[$\alpha$-mean-difference confidence interval]
\label{def:confIntv}
Given two continuous random variables $X_1 \sim \mathcal{D}_1$ and $X_2 \sim \mathcal{D}_2$ where $D_1,D_2$ are distributions, $Y=X_2-X_1$, and $\alpha \in [0,1]$. An interval $[l,u]$ is $\alpha$-mean-difference confidence interval if $P( l\leq \mu_Y \leq u )\geq 1-\alpha$.
\end{definition}

Now, we are ready to formalize \ddoip.

\begin{problem}[h!]
    \SetKwInOut{Input}{Input}
    \SetKwInOut{Output}{Output}
    \Input{A set  $\mathcal{S} = \{(x,c)\}$ s.t. $x$ is a realization of $X_c\sim D_c$, and $X_{c_1},X_{c_2}$ i.i.d. from the same $D_c$ if $c_1=c_2=c$. }
    \Output{Order of pairs of $D_{i} \preceq D_{j}$, and their $\alpha$-mean-difference confidence interval  $CI_{i,j}=[l_{i,j},u_{i,j}]$. }
    \caption{{\ddoip}}
	\label{prob1}
\end{problem}

\section{Statistical inference}
\label{sec:statInference}
 \subsection{Bootstrap approach}
 \label{sec:bootstrap}
Suppose we have $Y=X_2-X_1$ and $Y\sim D_Y$ with the unknown $\mu_Y$, we can use the mean $\bar{Y}=E[Y]$ as the point estimate of $\mu_Y$ since it is the unbiased estimator. We deploy the estimation statistics~\cite{cumming2013understanding,claridge2016estimation,ho2019moving}
, which is a framework that focuses on estimating an effect sizes, $Y$, of two distributions. Compared to  null hypothesis significance testing approach (NHST), estimation statistics framework reports not only whether two distribution are significantly different, but it also reports magnitudes of difference in the form of confidence interval. 

The estimation statistics framework uses Bootstrap technique~\cite{Efron1992} to approximately infer the bootstrap confidence interval of $\mu_Y$. Assuming that the number of times of bootstrapping is large, according to Central Limit Theorem (CLT), even though the underlying distribution is not normal distributed, summary statistics (e.g. means) of random sampling approaches a normal distribution.  Hence, we can use the normal confidence interval to approximate the confidence interval of $\mu_Y$.

\begin{theorem}[Central Limit Theorem (CLT)~\cite{Pishro-Nik2014introduction}]
Given $X_1,\dots,X_n$ be i.i.d. random variables with $E[X_i]=\mu <\infty$ and $0<\text{VAR}(X_i)=\sigma^2 < \infty$, and $\bar{X}=\frac{\sum_{i=1}^n X_i}{n}$. Then, the random variable

\begin{equation*}
    Z_n = \frac{\bar{X}-\mu}{\sigma/\sqrt{n}}
\end{equation*}

converges in distribution to a standard normal random variable as $n$ goes to infinity, 
that is 
\begin{equation*}
    \lim_{n\to\infty} P(Z_n\leq x) = \Phi(x), \forall x \in \mathbb{R},
\end{equation*}
where $\Phi(x)$ is the standard normal CDF.
\end{theorem}


\begin{lemma}
\label{lemma:confintv}
Given  $X_{1,1},\dots,X_{1,k}$ are random variables i.i.d. from $D_1$,  $X_{2,1},\dots,X_{2,k}$ are random variables i.i.d. from $D_2$, and $Y_1,\dots,Y_k$ are random variables where $Y_i=X_{2,i}-X_{1,i}$. 

Assuming that the number $k$ is large, the distribution of $Y_i$ is unknown with an unknown variance $\text{VAR}(Y_i)=\sigma^2_Y <\infty$.  Suppose $\bar{Y}$ is the sample mean of $Y_1,\dots,Y_k$, $\mu_Y=E[Y_i]$, and $s_{Y}$ is their standard deviation. Given that $\Phi(\cdot)$ is the standard normal CDF and $z_{\frac{\alpha}{2}}=\Phi^{-1}(1-\frac{\alpha}{2})$, then the interval

\begin{equation}
    CI_{\bar{Y}}=[\bar{Y} - z_{\frac{\alpha}{2} } \frac{s_{Y} }{\sqrt{k}}, \bar{Y} + z_{\frac{\alpha}{2} } \frac{s_{Y} }{\sqrt{k}} ]
\end{equation}

is approximately $(1-\alpha)100\%$ confidence interval for $\mu_Y$.
\end{lemma}
\begin{proof}
Since $k$ is large, the distribution of sample mean of $Y_1,\dots,Y_k$ follows the Central Limit Theorem. This implies that the random variable 

\begin{equation*}
    Z_k = \frac{\bar{Y}-\mu_Y}{\sigma_Y/\sqrt{k}}
\end{equation*}
has approximately $\mathcal{N}(0,1)$ distribution. Hence, $\bar{Y}$ is approximately normal distributed from $\mathcal{N}(\mu_Y,\sigma_Y/\sqrt{k})$. 
The $(1-\alpha)100\%$ confidence interval for $\bar{Y}$ is $[\mu_Y - z_{\frac{\alpha}{2} } \frac{\sigma_{Y}}{\sqrt{k}}, \mu_Y + z_{\frac{\alpha}{2} } \frac{\sigma_{Y}}{\sqrt{k}} ]$.

Since $\bar{Y}$ is the unbiased estimator of $\mu_Y$ and $s_{Y}$ is the unbiased estimator of $\sigma_{Y}$, we can have the approximation of $(1-\alpha)100\%$ confidence interval of $\mu_Y$ as follows. 

\begin{equation*}
    [\bar{Y} - z_{\frac{\alpha}{2} } \frac{s_{Y} }{\sqrt{k}}, \bar{Y} + z_{\frac{\alpha}{2} } \frac{s_{Y} }{\sqrt{k}} ]
\end{equation*}

\end{proof}

According to Lemma~\ref{lemma:confintv}, we need to access to a large number of $Y_1,\dots,Y_k$ to infer the confidence interval. We can generate $Y_1,\dots,Y_k$ s.t. $k$ is large using the bootstrap technique. The following theorem allows us to approximate the mean of $Y_i$ in the bootstrap approach.

\begin{theorem}[Bootstrap convergence~\cite{athreya1987bootstrap,bickel1981some}]
\label{thm:bootstrap}
Given $X_{1},\dots,X_{n}$ are random variables i.i.d. from an unknown distribution $D$ with $\text{VAR}(X_i)=\sigma^2 < \infty$. We choose $X'_1,\dots,X'_m$ from the set $\{X_{1},\dots,X_{n}\}$ by resampling with replacement. 
As $n,m$ approach $\infty$:

\squishlist
\item {\bf Asymptotic mean:} the conditional distribution of $\sqrt{m}(\bar{X}'-\bar{X})$ given $X_{1},\dots,X_{n}$ converges weakly to $\mathcal{N}(0,\sigma^2)$.
\item {\bf Asymptotic standard deviation:} $s_m \xrightarrow{} \sigma$ in conditional probability: that is for any positive $\epsilon$,
\begin{equation*}
 P(|s_m - \sigma|> \epsilon | X_{1},\dots,X_{n} ) \xrightarrow{} 0,
\end{equation*}
where $\bar{X}'=m^{-1}\sum_1^m X'_i$, $\bar{X}=n^{-1}\sum_1^n X_i$, and $s^2_m= m^{-1}\sum_1^m (X'_i-\bar{X}')^2$.
\squishend
\end{theorem}

From Theorem~\ref{thm:bootstrap}, when we increase the number of times we perform the resampling with replacement on $D_1,D_2$ to be large, we can approximate the $\bar{Y}$ using the bootstrap sample mean $\bar{Y}'$. The same applies for the standard deviation $s_{Y} $ that we can use its bootstrap version $s'_{Y}$ to approximate it.  By using $\bar{Y}',s'_{Y}$, we can approximate the confidence interval in Lemma~\ref{lemma:confintv}. 

 \subsection{Dominant-distribution relation inference}
 \label{sec:DDrelaionInfer}

According to Proposition~\ref{pros:DomExp}, $E[X_1]\leq E[X_2]$ implies $D_1 \preceq D_2$. Suppose that $\mu_1 = E[X_1]$ and $\mu_2=E[X_1]$ are also random variables.  If $P(\mu_1\leq \mu_2)$ or $P(\mu_2 - \mu_1 \geq 0) =1$, then $P(D_1 \preceq D_2) =1$. However, in reality,  $P(\mu_2 - \mu_1\geq 0)$ might not equal to one due to noise. Hence, we define the following notion of Dominant-distribution relation. 
 
 \begin{definition}[$\alpha$-Dominant-distribution relation]
 \label{def:EG-relation2}
Given two continuous random variables $X_1 \sim \mathcal{D}_1$ and $X_2 \sim \mathcal{D}_2$ where $D_1,D_2$ are distributions, and $\alpha \in [0,1]$. Suppose $\mu_1=E[X_1],\mu_2=E[X_2]$, we say that $D_2$ is dominant to $D_1$ if $P(E[\mu_2 - \mu_1]\geq 0) \geq 1-\alpha$; denoting $D_1 \preceq_\alpha D_2$.
 \end{definition}
 
 Suppose we have two empirical distribution $D'_1$ and $D'_2$. From Theorem~\ref{thm:bootstrap} and Lemma~\ref{lemma:confintv}, we can define $X_1$ and $X_2$ as random variables from sample-mean distributions $D_1,D_2$ of empirical distributions $D'_1$ and $D'_2$.  We can get $D_1$ and $D_2$ by bootstrapping data from $D'_1$ and $D'_2$. Suppose $Y=X_2 - X_1$, then, we can approximate the confidence interval of $\mu_Y = E[Y]$ with $\alpha$ using the interval $CI_{\bar{Y}}$ in Lemma $\ref{lemma:confintv}$.
 
 Next, we use $(1-\alpha)100\%$ confidence interval of $\mu_Y$ to infer whether $D_1 \preceq_\alpha D_2$. Given $\mu_y=\mu_2 - \mu_1$, according to the Definition~\ref{def:EG-relation2},  if $P(E[\mu_Y]\geq 0) \geq 1-\alpha$, then $D_1 \preceq_\alpha D_2$. We can approximate whether  $E[\mu_Y]\geq 0$ with the probability $1-\alpha$ by the approximate $(1-\alpha)100\%$ confidence interval of $\mu_Y$: $CI_{\bar{Y}}=[\bar{Y} - z_{\frac{\alpha}{2} } \frac{s_{Y} }{\sqrt{k}}, \bar{Y} + z_{\frac{\alpha}{2} } \frac{s_{Y} }{\sqrt{k}} ]$. If the lower bound $\bar{Y} - z_{\frac{\alpha}{2} } \frac{s_{Y} }{\sqrt{k}}$ is greater than zero, then $P(E[\mu_Y]\geq 0)$  is  approximately $1-\alpha$.
 
 In the aspect of hypothesis test, determining whether $D_1 \preceq_\alpha D_2$ is the same as testing whether the expectation of $X_1\sim D_1$ is less than the expectation of $X_2 \sim D_2$ where the null hypothesis is $E[X_2]-E[X_1]<0$ and the alternative hypothesis is $E[X_2]-E[X_1]\geq 0$. We can verify these two hypothesis by inferring the confidence interval of $\mu_Y=E[X_2]-E[X_1]$. If the lower bound of $\mu_Y$ is greater than zero with the probability $1-\alpha$, then we can reject the null hypothesis. Moreover, not only the confidence interval can test the null hypothesis, but it is also be able to tell us the magnitude of mean difference between $D_1$ and $D_2$. Hence, the confidence interval is more informative than the NHST approach. 
 
 Given a set of distributions $\{D_1,\dots,D_c\}$, in this paper, we choose to represent $\alpha$-Dominant-distribution relations using a network as follows.
 
 \begin{definition}[Dominant-distribution network]
 \label{def:DomDistNet}
Given a set of $c$ continuous distributions $S=\{D_1,\dots,D_c\}$ and $\alpha \in [0,1]$. Let $G=(V,E)$ be a directed acyclic graph. The graph $G$ is a Dominant-distribution network s.t. a node $i \in V$ represents $D_i$ and $(i,j)\in E$ if $D_j \preceq_\alpha D_i$.
 \end{definition}
 
 In the Section~\ref{sec:method}, we discuss about the proposed framework that can infer a Dominant-distribution network $G$ from a set of order-pairs of real value and category.